\documentclass[a4paper,fullpage]{article}

\usepackage{url}
\usepackage{fancyvrb}
\usepackage[all]{xy}
\SelectTips{cm}{}
\usepackage{amsthm}
\usepackage{authblk}

\newcommand{\sparql}{\mathrm{SPARQL}} 
\newcommand{\spa}{\mathrm{GrAL}} 
 
\newcommand{\empw}[1]{#1}

\newcommand{\A}{{\mathit{A}}} 
\newcommand{\V}{{\mathit{V}}}
\newcommand{\B}{{\mathit{B}}}
\newcommand{\I}{{\mathit{I}}} 
\newcommand{\IB}{{\mathit{IB}}}
\newcommand{\IBV}{{\mathit{IBV}}} 
\newcommand{\BV}{{\mathit{BV}}} 
\newcommand{\dom}{\mathit{Dom}}
\newcommand{\Card}{\mathit{Card}}
\newcommand{\To}{\Rightarrow}

\newcommand{\blank}[1]{\_\colon\!#1}
\newcommand{\ulm}{{\underline{m}}}  
\newcommand{\uln}{{\underline{n}}}
\newcommand{\ulp}{{\underline{p}}}
\newcommand{\im}{\mathit{Im}} 
\newcommand{\gr}{\mathit{Gr}} 
\newcommand{\true}{\mathit{true}} 
\newcommand{\false}{\mathit{false}} 
\newcommand{\Expr}{\mathit{Expr}} 
\newcommand{\Val}{\mathit{Val}} 
 
\newcommand{\ev}{\mathit{eval}}
\newcommand{\var}{\mathit{var}} 
\newcommand{\expr}{\mathit{expr}} 
\newcommand{\as}{\approx} 
\newcommand{\se}[1]{{[#1]}} 
\newcommand{\sem}[2]{{[[#1]]_{#2}}} 
\newcommand{\seg}[2]{{{#2}^{(#1)}}} 
\newcommand{\hsp}{\null\hspace{1.5pc}}
\newcommand{\construct}[2]{{\rm CONSTRUCT\;}#1{\rm\;WHERE\;#2}}
\newcommand{\select}[2]{{\rm SELECT\;}#1{\rm\;WHERE\;#2}}
\newcommand{\distinct}[2]{{\rm SELECT\;DISTINCT\;}#1{\rm\;WHERE\;#2}}

\newtheorem{example}{Example}
\newtheorem{remark}{Remark}
\newtheorem{definition}{Definition}
\newtheorem{note}{Notation}
\newtheorem{proposition}{Proposition}

\title{All You Need Is CONSTRUCT}
\date{}
\author[1]{Dominique Duval}
\author[2]{Rachid Echahed}
\author[2]{Fr\'ed\'eric Prost}
\affil[1]{LJK, Univ. Grenoble Alpes and CNRS, France} 
\affil[2]{LIG, Univ. Grenoble Alpes and CNRS, France}

\begin{document}

\maketitle

\begin{abstract}
In SPARQL, the query forms SELECT  and CONSTRUCT have been the subject
of  several studies,  both  theoretical and  practical.  However,  the
composition  of  such  queries  and their  interweaving  when  forming
involved  nested queries  has not  yet received  much interest  in the
literature.  We mainly  tackle the problem of  composing such queries.
For  this purpose,  we  introduce  a language  close  to SPARQL  where
queries can  be nested at  will, involving either CONSTRUCT  or SELECT
query forms and provide a formal  semantics for it.  This semantics is
based on a  uniform interpretation of queries. This  uniformity is due
to an extension of the notion  of RDF graphs to include isolated items
such  as  variables. As  a  key  feature of  this  work,  we show  how
classical SELECT queries can be easily encoded as a particular case of
CONSTRUCT queries.
\end{abstract}

\section{Introduction}
\label{intro}

Graph databases \cite{Robinson2013} are becoming ubiquitous in our
society. The success of this recent trend in the organization of data
stems from different scientific, technological and societal
factors. There are different ways to encode data in terms of graphs as proposed
 in the literature, see e.g., RDF graphs \cite{rdf} or
Property graphs \cite{RNpropertyGraphs}. Various query languages can
be associated to each data graph representation. In this paper, we
consider the W3C standards, namely RDF \cite{Robinson2013} formalism
to represent data graphs and its associated query language SPARQL
\cite{sparql}.

An RDF graph is defined as a set of \empw{RDF
triples}, where an RDF triple has the form
$(\empw{subject},\empw{predicate},\empw{object})$. The
\empw{subject} is either an \empw{IRI} (Internationalized Resource
Identifier) or a \empw{blank node}, the \empw{predicate} is an IRI and
the \empw{object} is either an IRI, a literal (denoting a value such as
a string, a number or a date) or a  \empw{blank node}. 

Notice that a predicate in an RDF triple cannot be a blank. For
example, a triple such as $(Paul, blank_{rel},Henry)$ standing
for ``there is some relationship between Paul and Henry''
is not allowed in RDF, but only in generalized RDF
\cite[Section~7]{rdf}. Following the theoretical point of view we
propose in this paper, there is no harm to consider blank
predicates within RDF triples. We thus consider \emph{data graphs} in a more
general setting including RDF graphs.


The language SPARQL, which is the standard query language associated
to RDF, features different query forms such as SELECT 
or CONSTRUCT forms, among others. Besides the W3C specifications of
SPARQL \cite{sparql}, different authors investigated formal semantics
of the language \cite{PerezAG09,Schmidt0L10}. The semantics
associated to SPARQL queries are not uniform in general. Indeed, for
instance the result of a SELECT query  is a multiset of mappings \cite{KKC}
while the result of a CONSTRUCT query
is a data graph \cite{Kostylevetal2015}. Because of these differences
between the semantics of the different forms of queries, building
nested queries becomes a bit cumbersome.

However, the need of nested queries as a feature of query languages is
well known \cite{Kim82} and nested SPARQL queries have already received
some interest in the litterature. For example, in \cite{KKC}, nesting
SELECT queries has been investigated throughoutly but CONSTRUCT
queries have not been considered ; while in \cite{PolleresRK16,AnglesG11},
either SELECT or CONSTRUCT queries can be nested but due to the chosen
semantics the FROM clause is required to nest CONSTRUCT queries.

In this paper, we consider query nesting for a core language close to
SPARQL and propose a new unified semantics for the main query forms
SELECT and CONSTRUCT. For this purpose, we clearly distinguish between
the evaluation of a query and its result. The evaluation of query over
a data graph is a set of mappings.  Here a mapping should be
understood algebraically as a graph homomorphism and not as a simple
assignement of variables.  The result of query is obtained by simply projecting
the right answer as a multiset of assignements of variables or as a
data graph according to the form of the query.
 From such semantics, one can compose queries of
 different forms to build involved nested subqueries.

\begin{example}
To illustrate briefly our proposal, we consider Example~1 in 
\cite{AnglesG11} and reformulate it in our framework without using FROM
clauses. In this example, one looks for emails of pairs of co-authors.

\begin{tabbing}
  012\=34\=56\=78\=90\=12\=34\=56\=7890123456789 \kill
  \> SELECT ?Mail1 ?Mail2 \\
  \>\>  WHERE \\
  \>\>\> $\{ \; \{$ CONSTRUCT $\{$ ?Aut1 co-author ?Aut2. $\}$ \\
  \>\>\>\> WHERE \\
  \>\>\>\>\>  $\{ \; \{$ ?Art bib:has-author ?Aut1 . ?Art bib:has-author
      ?Aut2.  $\}$  \\   
  \>\>\>\>\>   \phantom{$\{ \; \{$} FILTER ( !(?Aut1 = ?Aut2))   $\}$ \\
  \>\>\>\>\> $\}$  \\
  \>\>\> AND \\     
  \>\>\> $\{$ ?Per1 co-author ?Per2 .
    ?Per1 foaf:mbox ?Mail1 .\\
  \>\>\>  ?Per2 foaf:mbox ?Mail2 .
  $\}$ \\
\>\>\> $\}$

\end{tabbing}
\end{example}

In order to build such a uniform semantics we had to extend the notion
of RDF graphs to include isolated items.  As a key feature of this
work, we show how SELECT queries can be easily encoded as a
particular case of CONSTRUCT queries.


The paper is organized as follows. In the next section, we introduce
the main operators of a query graph algebra which are used later on
when investigating  the semantics of the proposed graph query
language. In Section~\ref{sec:spa}, a SPARQL-like language called
GrAL is introduced where queries are defined as specific  patterns.
This language is defined by its syntax and semantics together with
some illustrating examples. Concluding remarks and future work are given in
Section~\ref{conclusion}.

\section{The Graph Query Algebra} 
\label{sec:frame}

The Graph Query Algebra is a family of operations which are used in
Section~\ref{sec:spa} for defining the evaluation of queries in the 
Graph Algebraic Query Language $\spa$.
First mappings are introduced in Section~\ref{ssec:frame-graphs},
then operations for combining
sets of mappings are defined in Section~\ref{ssec:frame-operation}.
It is usual to describe the evaluation of queries in $\sparql$
in terms of mappings from variables to RDF terms, following \cite{PerezAG09}. 
In this paper, more precisely, we consider each mapping as a
morphism between graphs. 

\subsection{Sets of mappings}\
\label{ssec:frame-graphs}
 
\begin{definition}[graph on $\A$]
\label{def:frame-graph}
For any set $\A$ and any element $t=(s,p,o)$ in $\A^3$, 
the elements $s$, $p$ and $o$ are called respectively 
the \emph{subject}, the \emph{predicate} and the \emph{object} of $t$.
A \emph{graph} $X$ on $\A$ is made of a subset $X_N$ of $\A$
called the set of \emph{nodes} of $X$ and a subset $X_T$ of $\A^3$ 
called the set of \emph{triples} of $X$, such that the subject and the object
of each triple of $X$ is a node of $X$.
The nodes of $X$ which are neither a subject nor an object are called
the \emph{isolated nodes} of $X$.
The set of \emph{labels} of a graph $X$ on $\A$ is the subset $\A(X)$
of $\A$ made of the nodes and predicates of $X$.
\end{definition}

\begin{remark}
\label{rem:frame-graph}
Given two graphs $X_1$ and $X_2$ on $\A$ their \emph{union} $X_1\cup X_2$
is defined by $(X_1\cup X_2)_N = (X_1)_N \cup (X_2)_N$ and
$(X_1\cup X_2)_T = (X_1)_T \cup (X_2)_T$,
and similarly their \emph{intersection} $X_1\cap X_2$
is defined by $(X_1\cap X_2)_N = (X_1)_N \cap (X_2)_N$ and
$(X_1\cap X_2)_T = (X_1)_T \cap (X_2)_T$.
It follows that $\A(X_1\cup X_2)=\A(X_1) \cup \A(X_2)$ and
$\A(X_1 \cap X_2) = \A(X_1)\cap \A(X_2)$.
\end{remark}

\begin{definition}[morphism of graphs on $\A$]
\label{def:frame-morphism}
Let $X$ and $Y$ be graphs on a set $\A$.  
A \emph{morphism} $f$ (of graphs on $\A$) from $X$ to $Y$, denoted $f:X\to Y$, 
is a partial function from $\A(X)$ to $\A(Y)$
which \emph{preserves nodes} and \emph{preserves triples},
in the following sense.
Let $\dom(f)$ be the \emph{domain of definition} of $f$,
i.e., the subset of $\A(X)$ where the partial function $f$ is defined.
Then $f$ \emph{preserves nodes} if $f(n)\in Y_N$
for each $n \in X_N\cap\dom(f)$ 
and $f$ \emph{preserves triples} if $f^3(t)\in Y_T$ 
for each $t \in X_T\cap\dom(f)^3$.
Then $f_N:X_N\to Y_N$ and $f_T:X_T\to Y_T$ are the partial functions
restrictions of $f$ and $f^3$, respectively. 
Note that when $n$ is an isolated node of $X$ then the node $f(n)$ does not
have to be isolated in $Y$.
The \emph{domain} of a morphism $f:X\to Y$ is $X$ and its \emph{range} is $Y$.
A morphism $f:X\to Y$ \emph{fixes} a subset $C$ of $\A$
if $f(x)=x$ for each $x$ in $C\cap\A(X)$. 
Then the partial function $f$ is determined by its restriction to
$\A(X)\setminus C$. 
An \emph{isomorphism} of graphs on $\A$  
is a morphism $f:X\to Y$ of graphs on $\A$ that is invertible,
which means that both
$f_N:X_N\to Y_N$ and $f_T:X_T\to Y_T$ are bijections.
\end{definition}

\begin{definition}[image] 
\label{def:frame-image}
The \emph{image} of a graph $X$ by any partial function $f$ from
$\A(X)$ to $\A$ is the graph made of
the nodes $f(n)$ for $n\in X_N\cap\dom(f)$ and the triples $f^3(t)$ 
for $t\in X_T\cap\dom(f)^3$.
It is also called the image of $f$
and it is denoted either $\im(f)$ or $\im(X)$ when $f$ is clear
from the context.
Thus each partial function $f$ from $\A(X)$ to $\A$ is a morphism
of graphs on $\A$ from $X$ to $f(X)$.
\end{definition}

\begin{definition}[labels] 
\label{def:frame-labels} 
{F}rom now on,
the set $\A$ of labels of graphs is built from three disjoint countably
infinite sets $\I$, $\B$ and $\V$, 
called respectively the sets of \emph{resource identifiers}, 
\emph{blanks} and \emph{variables}.
We denote $\IB=\I\cup \B$, $\BV=\B\cup \V$ and $\IBV=\I\cup \B\cup \V$.
For each graph $X$ on $\IBV$ and each subset $Y$ of $\IBV$, 
the set $\A(X)\cap Y$ of labels of $X$ which belong to $Y$ is denoted $Y(X)$. 
\end{definition}

\begin{definition}[data and query graph, mapping]
\label{def:frame-mapping}
\emph{Data graphs} are finite graphs on $\IB$ and 
\emph{query graphs} are finite graphs on $\IBV$.
Thus each data graph can be seen as a query graph. 
A \emph{mapping} $m$ from a query graph $X$ to a data graph $Y$,
denoted $m:X\to Y$, 
is a morphism of query graphs from $X$ to $Y$ that fixes $\I$.
\end{definition}

\begin{remark}
\label{rem:frame-data-query} 
Intuitively, the resource identifiers are the ``constants'', 
that are fixed by morphisms, while both the blanks and variables are
the ``variables'', that may be instantiated.
It is only in construct queries (Section~\ref{ssec:spa-pattern})
that blanks and variables play truly distinct roles.
Thus, the precise symbol used for representing a blank or a variable does
not matter: a data graph is defined ``\emph{up to blanks}''
and a query graph ``\emph{up to blanks and variables}'',
and some care is required when several data or query graphs are in
the context. 
\end{remark}

\begin{remark}
\label{rem:frame-mapping}
Definition~\ref{def:frame-mapping}
means that a mapping $m:X\to Y$ is a partial function from $\IBV(X)$
to $\IB(Y)$ that fixes $\I$ and that preserves nodes and triples.
Thus, if there is a mapping from $X$ to $Y$ then $\I(X)\subseteq\I(Y)$. 
Each mapping $m:X\to Y$ determines a partial function $\mu:\BV\to\IB$,
defined by $\mu(x)=m(x)$ when $x\in\BV(X)$ and $\mu(x)$ is undefined
when $x\in\BV\setminus\BV(X)$. 
Conversely,
each partial function $\mu:\BV\to\IB$ can be extended as $\mu:\IBV\to\IB$
such that $\mu(x)=x$ for each $x\in\I$;
if $\mu:\IBV\to\IB$ preserves nodes and triples from $X$ to $Y$
then $\mu$ determines a mapping $m:X\to Y$,
defined by $m(x)=\mu(x)$ for each $x\in\IBV(X)$. 
In \cite{PerezAG09} and in subsequents papers like \cite{Kostylevetal2015,KKC}
a \emph{solution mapping}, or simply a \emph{mapping}, is a 
partial function $\mu:\V\to\IB$;
since it is assumed in these papers that patterns are blank-free, 
such mappings are related to our mappings in the same way as above,
by extending $\mu$ as $\mu:\IBV\to\IB$ by $\mu(x)=x$ for each $x\in\IB$.
\end{remark}

\begin{definition}[set of mappings]
\label{def:frame-set-of-mappings}
Let $X$ be a query graph and $Y$ a data graph. 
A \emph{set of mappings from $X$ to $Y$}, denoted $\ulm:X\To Y$,
is a finite set of mappings $m:X\to Y$. 
The \emph{domain} of $\ulm:X\To Y$ is $X$, its \emph{range} is $Y$,
and its \emph{image} is the subgraph of $Y$ union of the images of
the mappings in $\ulm$. 
\end{definition}

\begin{remark}[table] 
\label{rem:frame-table} 
A set of mappings $\ulm:X\To Y$ can be
represented as a \emph{table} $T(\ulm)$ made of one line
for each mapping $m$ in $\ulm$ and one column for each $x\in\BV(X)$, 
with the entry in line $m$ and column $x$ equal to $m(x)\in\IB(Y)$ when it is
defined and $\bot$ otherwise.
The order of the rows and columns of $T(\ulm)$ is arbitrary.
Note that $\ulm$ is determined by the table $T(\ulm)$ together with
$X$ and $Y$, but in general $\ulm:X \To Y$ cannot be recovered from
$T(\ulm)$ alone.

\end{remark}

\begin{definition}[compatible mappings] 
\label{def:frame-compatible}
Two mappings $m_1:X_1\to Y_1$ and $m_2:X_2\to Y_2$ are \emph{compatible},
written as $m_1\sim m_2$, if $m_1(x)=m_2(x)$ for each $x\in\BV(X_1)\cap\BV(X_2)$.
This means that for each $x\in \BV(X_1)\cap\BV(X_2)$, 
$m_1(x)$ is defined if and only if $m_2(x)$ is defined and then $m_1(x)=m_2(x)$. 
Given two compatible mappings $m_1:X_1\to Y_1$ and $m_2:X_2\to Y_2$, 
there is a unique mapping
$m_1\bowtie m_2 : X_1\cup X_2 \to Y_1\cup Y_2$ 
such that $m_1\bowtie m_2\sim m_1$ and $m_1\bowtie m_2\sim m_2$,
which means that $m_1\bowtie m_2$ coincides with $m_1$ on $X_1$
and with $m_2$ on $X_2$. 
\end{definition}

\begin{remark}[About RDF and SPARQL] 
\label{rem:frame-rdf}
When dealing with RDF and SPARQL \cite{rdf,sparql} 
the set $\I$ is the disjoint union of the set
of \emph{IRI}s (Internationalized Resource
Identifiers) and the set of \emph{literal}s. 
An \emph{RDF graph} is a set of triples on $\IB$, that is, a graph on $\IB$
without isolated node, 
where all predicates are IRIs and only objects can be literals.
Thus an \emph{isomorphism} of RDF graphs, as defined in \cite{rdf},
is an isomorphism of graphs on $\IB$ as in Definition~\ref{def:frame-morphism}.
The set of \emph{RDF terms} of an RDF graph $X$ is the set $\IB(X)$. 
Similarly a \emph{basic graph pattern} of SPARQL is a set of triples on $\IBV$ 
where all predicates are IRIs or variables and only objects can be literals.
Thus data graphs and query graphs generalize RDF graphs
and basic graph patterns, respectively. 
\end{remark} 

\subsection{Operations on sets of mappings}\  
\label{ssec:frame-operation}

In this Section we define some elementary transformations between
sets of mappings.

\begin{remark}[expressions and values]
\label{rem:frame-expr}  
We assume the existence of a set $\Expr$ of \emph{expressions}
with subsets $\V(\expr)$ of $\V$ and $\B(\expr)$ of $\B$ for each
expression $\expr$.
For each query graph $X$ the \emph{expressions on $X$}
are the expressions $\expr$ such that $\V(\expr)\subseteq \V(X)$ and
$\B(\expr)\subseteq \B(X)$.
We assume that there is a subset $\Val$ of $\I\cup\{\bot\}$
called the set of \emph{values}
and that for each mapping $m:X\to Y$ and each expression $\expr$ on $X$ 
there is a value $m(\expr)\in\Val$.
We assume that the boolean values $\true$ and $\false$ are in $\Val$, 
as well as the numbers and strings. 
\end{remark}

The first transformation on sets of mappings is the fundamental join operation. 

\begin{definition}[join]
\label{def:frame-join}
Given two sets of mappings $\ulm_1:X_1\To Y_1$ and $\ulm_2:X_2\To Y_2$, 
the \emph{join} of $\ulm_1$ and $\ulm_2$ is the 
set of mappings $\mathit{Join}(\ulm_1,\ulm_2) : X_1\cup X_2 \To Y_1\cup Y_2$
made of the mappings $m_1\bowtie m_2$ for all compatible mappings 
$m_1\in\ulm_1$ and $m_2\in\ulm_2$:
\\ \hsp $\mathit{Join}(\ulm_1,\ulm_2) = 
\{ m_1\bowtie m_2 \mid m_1\in\ulm_1 \wedge m_2\in\ulm_2 \wedge m_1\sim m_2\} :
X_1\cup X_2 \To Y_1\cup Y_2$.
\end{definition}

Subsets of a set of mappings can be defined by a filter operation.

\begin{definition}[filter]
\label{def:frame-filter}
Let $\ulm:X\To Y$ be a set of mappings and let $\expr$ be an expression on $X$. 
The \emph{filter of $\ulm$ by $\expr$} 
is the set of mappings $m$ in $\ulm$ where $m(\expr)=\true$:
\\ \hsp $\mathit{Filter}(\ulm,\expr) =
\{ m \mid m\in\ulm \wedge m(\expr)=\true\} :
X \To Y$. 
\end{definition}

Given a mapping $m:X\To Y$ and a query graph $X'$ contained in $X$,
let $m|_{X'}$ denote the restriction of $m$ to $X'$. 

\begin{definition}[restriction]
\label{def:frame-restriction}
Given a set of mappings $\ulm:X\To Y$ and a query graph $X'$ contained in $X$, 
the \emph{restriction of $\ulm$ to $X'$} is the set of 
mappings $\mathit{Restrict}_{X'}(\ulm): X' \to Y$ 
made of the restrictions $m|_{X'}$ of the mappings $m$ in $\ulm$ to $X'$:
\\ \hsp $\mathit{Restrict}(\ulm,X') = \{m|_{X'} \mid m\in\ulm \} : X'\To Y$.
\\
Since different mappings in $\ulm$ may coincide on $X'$,
the number of mappings in $\mathit{Restrict}_{X'}(\ulm)$ may be smaller
than the number of mappings in $\ulm$.
\end{definition}

\begin{note}[Notation] 
Given a mapping $m:X\to Y$ and a query graph $X'$ containing $X$,
there are several ways to extend $m$ as $m':X'\to Y\cup \im(X')$
where $\im(X')$ is the data graph image of $X'$ by $m'$. 
For instance, depending on the kind of labels in
$D = \IBV(X')\setminus\IBV(X)$: 
\begin{itemize}
\item 
For any $D$, $m$ can be extended as $m':X'\to Y\cup \im(m')$ such that 
$m'(x)=x$ for each $x\in D\cap\I$ and 
$m'(x)$ is undefined (denoted $m'(x)=\bot$) for each $x\in D\cap\BV$. 
This is denoted: 
\\ \hsp $m'=\mathit{Ext}_{\bot}(m,X') : X'\to Y\cup \im(m')$.
\item 
If $D \subseteq \IB$ then $m$ can be extended as
$m':X'\to Y\cup \im(m')$ such that 
$m'(x)=x$ for each $x\in D\cap\I$ and 
$m'(x)$ is a fresh blank for each $x\in D\cap\B$. 
This is denoted: 
\\ \hsp $m'=\mathit{Ext}_{\IB}(m,X') : X'\to Y\cup \im(m')$.
\item  
If $D$ is made of one variable $\var$ 
and $\expr$ is an expression on $X$ then $m$ can be extended as
$m':X'\to Y\cup\{m(\expr)\}$ such that $m'(\var)=m(\expr)$.
This is denoted: 
\\ \hsp $m'=\mathit{Ext}_{\var\as\expr}(m,X') : X'\to Y\cup\{\expr\}$. 
\end{itemize}
\end{note}

\begin{definition}[extension] 
\label{def:frame-extension}
Given a set of mappings $\ulm:X\To Y$ and a query graph $X'$ containing $X$,
let $D = \IBV(X')\setminus\IBV(X)$. We define the following extensions of
$\ulm$ as $\ulm':X'\To Y\cup \im(X')$ where $\im(X')= Y\cup \ulm'(X')$: 
\begin{itemize}
\item 
The \emph{extension of $\ulm$ to $X'$ by undefined functions} is: 
\\ \hsp $\mathit{Extend}_{\bot}(\ulm,X') = \{\mathit{Ext}_{\bot}(m,X')
\mid m\in\ulm\} : X'\To Y'$.
\item 
If $D \subseteq \IB$ then
the \emph{extension of $\ulm$ to $X'$ by fresh blanks} is: 
\\ \hsp $\mathit{Extend}_{\IB}(\ulm,X') = \{\mathit{Ext}_{\IB}(m,X')
\mid m\in\ulm\} : X'\To Y' $.
\item 
If $D=\{\var\}$ for a variable $\var$ 
and $\expr$ is an expression on $X$ then the
\emph{extension of $\ulm$ to $X'$ by binding $\var$ to the values of $\expr$}
is: 
\\ \hsp $\mathit{Extend}_{\var\as\expr}(\ulm,X') =
\{\mathit{Ext}_{\var\as\expr}(m,X')
\mid m\in\ulm\} : X'\To Y\cup\{\expr\}$. 
\end{itemize}
Note that the number of mappings in any extension of $\ulm$ is the same as
in $\ulm$.
\end{definition}

For defining the union of two sets of mappings, we first extend them
by undefined functions in such a way that they both get the same domain
and range.  

\begin{definition}[union]
\label{def:frame-union}
The \emph{union} $\mathit{Union}(\ulm_1,\ulm_2): X_1\cup X_2 \To Y_1\cup Y_2$
of two sets of mappings
$\ulm_1:X_1\To Y_1$ and $\ulm_2:X_2\To Y_2$ is the set-theoretic union of
their extensions to $X_1\cup X_2$ by undefined functions:
\\ \hsp $\mathit{Union}(\ulm_1,\ulm_2) = 
\mathit{Extend}_{\bot}(\ulm_1,X_1\cup X_2) \,\cup\,
\mathit{Extend}_{\bot}(\ulm_2,X_1\cup X_2) 
: X_1\cup X_2 \To Y_1\cup Y_2 $.
\end{definition}

Finally, we will use the well-known \emph{projection} operation
for building a multiset of mappings from a set of mappings.

\begin{definition}[projection]
\label{def:queries-projection}
The \emph{projection} of a set of mappings $\ulm:X\To Y$ 
to a subgraph $X'$ of $X$ is the \emph{multiset} of
mappings $\mathit{Project}(\ulm,X')$ with base set
$\mathit{Restrict}(\ulm,X'):X'\To Y$ and with multiplicity for each mapping
$m'$ the number of mappings $m\in\ulm$ such that $m'=m|_{X'}$.  
Thus the number of mappings in $\mathit{Project}(\ulm,X')$,
counting multiplicities, is always the same
as the number of mappings in $\ulm$.
\end{definition}

\section{The Graph Algebraic Query Language} 
\label{sec:spa}

In this Section we introduce the Graph Algebraic Query Language $\spa$.
Its syntax and semantics for expressions and patterns 
are defined in a mutually recursive way:
this is mainly due to the fact that expressions can be defined
from patterns, using the EXISTS and NOT EXISTS syntactic blocks. 
Syntactically,
the \emph{queries} of $\spa$ are seen as patterns from the beginning:
a query is a specific kind of pattern.
Semantically,
the \emph{value} of a pattern over a data graph is a set of mappings
(Section~\ref{ssec:spa-pattern}).
In addition, when a pattern is a query then its \emph{result}
can be derived from its value:
the result of a contruct-query is a data graph,
the result of a select-distinct-query is a set of mappings,
and the result of a select-query is a multiset of mappings,
as in $\sparql$ (Section~\ref{ssec:spa-query}).

\subsection{Expressions, patterns and queries}\  
\label{ssec:spa-pattern} 

A \emph{basic expression} is defined as usual from constants
(numbers, strings, boolean values) and variables (and blanks,
which act as variables here), using formal operations
like $+$, $-$, $\mathit{concat}$, $>$, $\wedge$,...
The \emph{basic expressions on $X$} are defined as in
Remark~\ref{rem:frame-expr}. 

\begin{definition}[Syntax of expressions]
\label{def:spa-syn-expr}
An \emph{expression} $\expr$ in the language $\spa$ is either a
basic expression or an expression of the form
EXISTS $P_1$ or NOT EXISTS $P_1$ for some pattern $P_1$, 
which are expressions on $X$ for every query graph $X$.
\end{definition}

\begin{definition}[Syntax of patterns]
\label{def:spa-syn-pattern} 
A \emph{pattern} $P$ in the language $\spa$ is defined inductively as follows.
\begin{itemize}
\item A query graph is a pattern, called a \emph{basic pattern}. 
\item If $P_1$ and $P_2$ are patterns then the following are patterns: 
\begin{itemize}
\item[] $P_1{\rm\;AND\;}P_2$
\item[] $P_1{\rm\;UNION\;}P_2$
\end{itemize}
\item If $P_1$ is a pattern and $\expr$ an expression on $P_1$ 
then the following is a pattern: 
\begin{itemize}
\item[] $P_1{\rm\;FILTER\;}\expr$ 
\end{itemize}
\item If $P_1$ is a pattern, $\expr$ an expression on $P_1$ 
and $\var$ a fresh variable then the following is a pattern: 
\begin{itemize}
\item[] $P_1{\rm\;BIND\;}(\expr{\rm\;AS\;}\var)$
\end{itemize}
\item If $P_1$ is a pattern and $R$ a query graph
such that $\V(R)\subseteq\V(P_1)$ 
then the following is a pattern: 
\begin{itemize}
\item[] $\construct{R}{P_1}$
\end{itemize}
\item If $P_1$ is a pattern and $S$ a finite set of variables 
such that $S\subseteq\V(P_1)$ 
then the following are patterns: 
\begin{itemize}
\item[] $\distinct{S}{P_1}$
\item[] $\select{S}{P_1}$
\end{itemize}
\end{itemize}
\end{definition}

The semantics of
expressions and patterns are defined in a mutually recursive way.
The value of an expression $\expr$ on $X$ with respect to
a set of mappings $\ulm:X\To Y$ is a family
$\ev(\ulm,\expr)= (m(\expr))_{m\in\ulm}$ of elements of $\Val$
(Definition~\ref{def:spa-sem-expr}). 
The value of a pattern $P$ over 
a data graph $G$ is a set of mappings $\sem{P}{G}:\se{P}\To\seg{P}{G}$
from a query graph $\se{P}$ depending only on $P$ 
to a data graph $\seg{P}{G}$ that contains $G$
(Definitions~\ref{def:spa-sem-pattern} and~\ref{def:spa-sem-query}).

\begin{definition}[Evaluation of expressions]
\label{def:spa-sem-expr}  
The \emph{value} of an expression $\expr$ on $X$ with respect to
a set of mappings $\ulm:X\To Y$ is the family
$\ev(\ulm,\expr)= (m(\expr))_{m\in\ulm}$ of elements of $\Val$
defined as follows: 
\begin{itemize}
\item If $\expr$ is a basic expression then $m(\expr)$ is the given value of
$\expr$ with respect to $m$.
\item If $\expr={\rm EXISTS\;}P_1$ then 
$m(\expr)$ is $\true$ if there is some $m_1\in\sem{P_1}{G}$ 
such that $m \sim m_1 $ and $\false$ otherwise.
\item If $\expr={\rm NOT\;EXISTS\;}P_1$ then $m(\expr)$ is the negation of
$m({\rm EXISTS\;}P_1)$.
\end{itemize}
\end{definition}

\begin{definition}[Equivalence of patterns]
\label{def:spa-equiv}
Two patterns are \emph{equivalent} if they have the same value
over $G$ for every data graph $G$, up to a renaming of blanks.
\end{definition}

\begin{definition}[Evaluation of non-query patterns]
\label{def:spa-sem-pattern}
The \emph{value} of a pattern $P$ of $\spa$ over 
a data graph $G$ is a set of mappings $\sem{P}{G}:\se{P}\To\seg{P}{G}$
from a query graph $\se{P}$ depending only on $P$
to a data graph $\seg{P}{G}$ that contains $G$.
Below is the first part of the recursive definition of the value 
of $P$ over $G$, the second part is given in Definition~\ref{def:spa-sem-query}.
\begin{itemize}

\item 
If $P$ is a basic pattern then $\se{P}=P$, $\seg{P}{G}=G$ and 
\\ \hsp $\sem{P}{G} : P\To G$ is the set of 
all total mappings from $P$ (as a query graph) to $G$.

\item If $P_1$ and $P_2$ are patterns then
\\ \hsp $\sem{P_1 {\rm \;AND\;}P_2}{G} =
\mathit{Join}(\sem{P_1}{G},\sem{P_2}{\seg{P_1}{G}}) :
\se{P_1}\cup\se{P_2} \To \seg{P_2}{(\seg{P_1}{G})}$.

\item If $P_1$ and $P_2$ are patterns then
  \\ \hsp $\sem{P_1 {\rm \;UNION\; } P_2}{G} =
  \mathit{Union} (\sem{P_1}{G}, \sem{P_2}{\seg{P_1}{G}}) :
\se{P_1}\cup\se{P_2} \To \seg{P_2}{(\seg{P_1}{G})}$. 

\item If $P_1$ is a pattern and $\expr$ an expression on $P_1$ then
\\ \hsp $\sem{P_1{\rm \;FILTER\; } \expr}{G} =
\mathit{Filter} (\sem{P_1}{G},\expr) :
\se{P_1} \To \seg{P_1}{G} $. 

\item If $P_1$ is a pattern, $\expr$ an expression on $P_1$ 
and $\var$ a fresh variable then 
\\ \hsp $\sem{P_1{\rm \;BIND\; }(\expr{\rm \;AS\; }\var)}{G} =
\mathit{Extend}_{\var\as\expr} (\sem{P_1}{G},\se{P_1} \cup \{\var\}) : $
\\ \hsp\hsp $\se{P_1} \cup \{\var\} \To
\seg{P_1}{G}\cup \{m(\expr)\mid m\in\sem{P_1}{G}\}$. 

\end{itemize}
\end{definition}

Definition~\ref{def:spa-sem-pattern} and Remark~\ref{rem:spa-sem-pattern}
are illustrated by Examples~\ref{ex:union} to~\ref{ex:bind}.

\begin{remark}
\label{rem:spa-sem-pattern}
  
Whenever $\sem{P_1}{G}=G$ then $\sem{P_1 {\rm \;AND\; }P_2}{G}$ 
and $\sem{P_1 {\rm \;UNION\; } P_2}{G}$ are symmetric in $P_1$ and $P_2$.
This is the case when the pattern $P$ contains no BIND, 
CONSTRUCT, SELECT DISTINCT or SELECT.
In particular, a pattern composed of basic patterns related by ANDs
is equivalent to the basic pattern union of its components.
But in general the data graph $\seg{P_1}{G}$ may be strictly larger than $G$, 
so that the semantics of $P_1{\rm \;AND\; }P_2$ and
$P_1{\rm \;UNION\; }P_2$ is not symmetric in $P_1$ and $P_2$.
The semantics of patterns in $\spa$ is a set semantics:
each set of mappings $\sem{P_1 {\rm \;UNION\; } P_2}{G}$
is a set, not a multiset.
However for select-queries it is possible to keep the multiplicities,
as explained in Remark~\ref{rem:spa-query-union}. 

The value of $P_1 {\rm \;FILTER \;EXISTS\; } P_2$ can be expressed
without mentioning expressions. Indeed, it follows from
Definition~\ref{def:spa-sem-pattern} that 
\\ \hsp $\sem{P_1 {\rm \;FILTER \;EXISTS\; } P_2}{G} =  
\mathit{Restrict}(\mathit{Join}(\sem{P_1}{G},\sem{P_2}{\seg{P_1}{G}}),\se{P_1})$.

In order to evaluate  $P_1{\rm \;BIND\; }(\expr{\rm \;AS\; }\var)$ over $G$,
the fresh variable $\var$ is added to the query graph $\se{P_1}$ as
an isolated node and the values $m(\expr)$ are added to the data graph
$G$ as nodes, which are isolated if they are not yet nodes of $G$. 

\end{remark}

\begin{definition}[Evaluation of query patterns]
\label{def:spa-sem-query}
Below is the second part of the recursive definition of the value 
of a pattern $P$ of $\spa$ over a data graph $G$. 
The first part is given in Definition~\ref{def:spa-sem-pattern}.

\begin{itemize}

\item If $P_1$ is a pattern and $R$ a query graph
such that $\V(R)\subseteq\V(P_1)$ 
then 
\\ \hsp $\sem{\construct{R}{P_1}}{G} =
\mathit{Restrict}( \mathit{Extend}_{\IB}(\sem{P_1}{G},\se{P_1}\cup R) , R) :
\\ \hsp\hsp R\To  \seg{P_1}{G} \cup \im(R)$.

\item If $P_1$ is a pattern and $S$ a finite set of variables 
such that $S\subseteq\V(P_1)$ 
then 
\\ \hsp $\sem{\distinct{S}{P_1}}{G} = 
\mathit{Restrict}(\sem{P_1}{G},S) :
\\ \hsp\hsp S\To G \cup \im(S)$.

\item If $P_1$ is a pattern and $S$ a finite set of variables
such that $S\subseteq\V(P_1)$ 
let $\gr(S)$ denote the query graph made of a fresh blank node $s$
and a triple $(s,p_\var,\var)$ for some chosen element $p_\var$ of $\I$ 
for each variable $\var$ in $S$, then 
\\ \hsp $\sem{\select{S}{P_1}}{G} = 
\mathit{Restrict}(\mathit{Extend}_{\IB}(\sem{P_1}{G},\se{P_1}\cup\gr(S)),\gr(S)):
\\ \hsp\hsp \gr(S)\To G \cup \im(\gr(S))$

\end{itemize}

\end{definition}

Definition~\ref{def:spa-sem-query} and Remark~\ref{rem:spa-sem-query}
are illustrated by Examples~\ref{ex:construct} to~\ref{ex:select}.

\begin{remark}
\label{rem:spa-sem-query}

In order to evaluate $Q=\construct{R}{P}$ over $G$  
one has to look for the mappings of $P$ in $\sem{P}{G}$,
then build a copy of $R$ for each such mapping and
finally merge these copies by duplicating in a suitable way the blanks of $R$.
The construction of the set of mappings $\ulp=\sem{Q}{G}: R \To \seg{Q}{G}$
from $\ulm=\sem{P}{G}:\se{P} \To \seg{P}{G}$ can be described as follows.
First a family of renaming functions $(d_m)_{m\in\ulm}$ is built,
such that each $d_m$ is an injective function from 
$\B(R)$ to the set of blanks which are fresh, i.e.,
the blanks which are not used anywhere in the context
(thus, specifically, not in $G$), 
and the functions $d_m$ have pairwise disjoint images.
For each $m$ the function $d_m$ is used for extending $m$ as the unique
mapping $n$ on $\se{P}\cup R$ such that $n(x)=m(x)$ for each $x\in\se{P}$ 
and $n(x)=d_m(x)$ for each $x\in\B(R)$. 
Then $\uln$ is restricted as $\ulp$ with domain $R$ 
by restricting each $n\in\uln$ to the subgraph $R$ of $\se{P}\cup R$.

\begin{figure}[ht]
$$\xymatrix@C=3pc{
  \ar@{}[rd]|{\mathit{Extend}_\IB}
  \se{P} \ar@{}[r]|{\subseteq} \ar@{=>}[d]_{\ulm} &
  \ar@{}[rd]|{\mathit{Restrict}}
  \se{P}\cup R \ar@{=>}[d]_{\uln} &
  R \ar@{}[l]|{\supseteq} \ar@{=>}[d]_{\ulp} \\
  \seg{P}{G} \ar@{}[r]|{\subseteq} &
  \seg{Q}{G} & \seg{Q}{G} \ar@{}[l]|{=} \\
}$$
\caption{Evaluation of a construct query.}
\label{fig:sem-construct}  
\end{figure}

For select-distinct queries, Definition~\ref{def:spa-sem-query} implies that
\\ \hsp\hsp $\distinct{S}{P_1} \equiv \construct{S}{P_1} $
\\ Indeed, the set of variables $S$ can be seen as a query graph
made of isolated nodes, all of them variables. Then the set 
$\IB(S)$ is empty and consequently $\uln = \ulm$:
the extension step is useless. 

For select-queries, Definition~\ref{def:spa-sem-query} implies that
\\ \hsp\hsp $\select{S}{P_1} \equiv \construct{\gr(S)}{P_1}$
\\ The set $\IB(\gr(S))$ is non-empty:
there is one blank $s$ in $\gr(S)$ and one element $p_\var$ of $\I$ for each
element $\var$ of $S$.
It follows that $\uln$ extends each $m\in\ulm$ with a fresh blank,
image of $s$, 
which can be seen as an identifier for each $m\in\ulm$.
When restricting $\uln$ for computing $\ulp$ this identifier is kept,
so that all mappings remain distinct.

\end{remark}

\subsection{Queries: value and result}\ 
\label{ssec:spa-query} 

Definition~\ref{def:spa-syn-pattern} says that each query is a pattern
and Definition~\ref{def:spa-sem-query} says that the value of a query is
its value as a pattern, so that it is always a set of mappings,
whatever the query form is.
But the \emph{result} of a query, which is defined as a by-product of its
evaluation, does depend on the query form:
it is a data graph for construct-queries, a set of mappings for
select-distinct-queries and a multiset of mappings for select-queries.

\begin{definition}[syntax of queries]
\label{def:spa-query-syn}
A \emph{query} in the language $\spa$ is a pattern of one of the following
forms, where $P$ is a pattern, $R$ a query graph and $S$ a finite set of
variables.
\begin{itemize}
\item $\construct{R}{P}$
\item $\distinct{S}{P}$
\item $\select{S}{P}$
\end{itemize}
\end{definition}

The value of a query $Q$ over a data graph $G$ in the language
$\spa$ is its value as a pattern (Definition~\ref{def:spa-sem-query}), 
it is the set of mappings $\sem{Q}{G} : \se{Q} \To \seg{Q}{G}$.
In addition, each query $Q$ has a \emph{result} over $G$,
which is defined below from its value $\sem{Q}{G}$ in a way
that depends on the form of the query.

\begin{definition}[result of queries] 
\label{def:spa-query-sem}
The \emph{result} of a query $Q$ over a data graph $G$
is defined from the value $\sem{Q}{G}: \se{Q} \To \seg{Q}{G}$ as follows:
\begin{itemize}
\item If $Q=\construct{R}{P}$ its result is the \emph{data graph}
image of $\se{Q}$ by $\sem{Q}{G}$: 
\\ \hsp  $\mathit{Result}(Q,G) = \im(\sem{Q}{G}) $. 
\item If $Q=\distinct{S}{P}$ its result is the \emph{set of mappings}
$\sem{Q}{G}$: 
\\ \hsp  $\mathit{Result}(Q,G) = \sem{Q}{G}:S \To \seg{Q}{G}$. 
\item If $Q=\select{S}{P}$ its result is the \emph{multiset of mappings}
projection of $\sem{Q}{G}$: 
\\ \hsp $\mathit{Result}(Q,G) = \mathit{Project}(\sem{Q}{G},S):
S \To \seg{Q}{G}$. 
\end{itemize}
\end{definition}

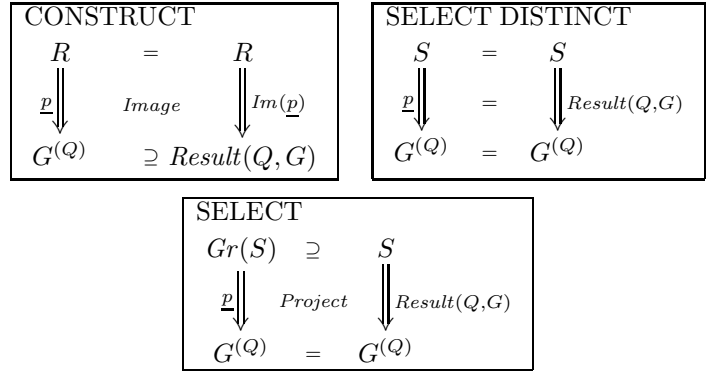
\begin{figure}[ht]
$$ \begin{array}{|l|l|l|}
  \cline{1-1}\cline{3-3}
{\rm CONSTRUCT} && {\rm SELECT\;DISTINCT}  \\ 
\xymatrix@C=2pc{
  \ar@{}[rd]|{\mathit{Image}}
  R \ar@{=>}[d]_{\ulp} &
  R \ar@{}[l]|{=} \ar@{=>}[d]^{\im(\ulp)} \\
  \seg{Q}{G} & 
  \mathit{Result}(Q,G) \ar@{}[l]|{\supseteq} \\
}
& \; &
\xymatrix@C=2pc{
  \ar@{}[rd]|{\qquad=\qquad}
  S \ar@{=>}[d]_{\ulp} & 
  S \ar@{}[l]|{=} \ar@{=>}[d]^{\mathit{Result}(Q,G)} \\
  \seg{Q}{G} & \seg{Q}{G} \ar@{}[l]|{=} \\ 
}\\
\cline{1-1}\cline{3-3}
\end{array}$$
$$
\begin{array}{|l|}
\cline{1-1}
{\rm SELECT}\\ 
\xymatrix@C=2pc{
  \ar@{}[rd]|{\mathit{Project}}
  \gr(S) \ar@{=>}[d]_{\ulp} &
  S \ar@{}[l]|{\supseteq} \ar@{=>}[d]^{\mathit{Result}(Q,G)} \\
  \seg{Q}{G} & \seg{Q}{G} \ar@{}[l]|{=} \\ 
}\\
  \cline{1-1}
\end{array}
  $$
  \caption{Result of queries.}
\label{fig:sem-result}  
\end{figure}

\begin{remark}[RDF and SPARQL] 
\label{rem:spa-query-rdf}
When $Q$ is a construct-query and the data graph $G$ is an RDF graph,
it may happen that the data graph $\mathit{Result}(Q,G)$ is not an RDF graph.
But the largest RDF graph included in $\mathit{Result}(Q,G)$
is the \emph{answer} to $Q$ over $G$ in the sense of
\cite[Section~5]{Kostylevetal2015}: this derives from the description
of $\mathit{Result}(Q,G)$ in Remark~\ref{rem:spa-sem-query}.
Using this Remark~\ref{rem:spa-sem-query} we also get a description of
the result of select-distinct-queries and select-queries that is the same as
in \cite[Section~2.3]{KKC}: 
For select-distinct-queries, the result is the set of mappings 
which consists of the restrictions of all mappings in $\sem{P}{G}$.
For select-queries, the result is the multiset of mappings
with base set the restrictions of all mappings $m$ in $\sem{P}{G}$,
each one with multiplicity the corresponding number of $m$'s. 
\end{remark}

\begin{proposition}[value]
\label{prop:spa-query-value}
For any query $Q$ with pattern $P$ and any data graph $G$, 
the number of mappings in $\sem{Q}{G}$ cannot be larger than
the number of mappings in $\sem{P}{G}$. 
\end{proposition}

\begin{proof}
Since select-queries and select-distinct-queries are equivalent 
to construct-queries with the same pattern, 
we may assume that $Q=\construct{R}{P}$
for a pattern $P$ and a query graph $R$. 
With the notations $\ulm=\sem{P}{G}$,
$\uln=\mathit{Extend}_{\IB}(\ulm,\se{P}\cup R)$
and $\ulp=\mathit{Restrict}(\uln,R)$, so that $\ulp=\sem{Q}{G}$, 
we know from Definitions~\ref{def:frame-restriction}
and~\ref{def:frame-extension} that $\Card(\ulp)\leq\Card(\uln)=\Card(\ulm)$.
\end{proof}

\begin{remark}[result]
\label{rem:spa-query-result}
In general the value of a construct-query 
cannot be deduced from its result alone.
However, for select-distinct-queries the value is the result
and for select-queries the value may be recovered fom the result
by choosing any fresh blanks as the images of the unique blank of $\gr(S)$.
\end{remark}

\begin{remark}[UNION and UNION ALL]
\label{rem:spa-query-union}
The union of two multisets $M_1$ and $M_2$, respectively based on the
sets $X_1$ and $X_2$, is usually defined as the multiset
$M$ based on the set $X_1 \cup X_2$ where the multiplicity of each element
is the sum of its multiplicities in $M_1$ and $M_2$.
When dealing with select-queries, the keyword UNION is used in $\sparql$
for the union as multisets. In $\mathrm{SQL}$
the union as multisets is obtained via the keyword UNION ALL, 
while UNION returns the union of the base sets. 
In $\spa$, the keyword UNION always returns a set of mappings.
In order to get the union as multisets of mappings we define UNION ALL
as follows,
with $S=\V(P_1)\cup \V(P_2)$: 
\\ \hsp $P_1 {\rm \;UNION\; ALL\; } P_2 = \\ \hsp \hsp
\{\select{S}{P_1}\} {\rm \;UNION\; } \{\select{S}{P_2}\} $.
\\ See Examples~\ref{ex:union} and~\ref{ex:union-all}. 
\end{remark}
  
\begin{remark}[subqueries]
\label{rem:spa-query-sub}
Since queries are specific patterns, they can be combined at will
between themselves and with other patterns, using the various 
syntactic building blocks for getting patterns. 
In particular, this provides various kinds of subqueries.
See Example~\ref{ex:sub}. 
Note that for computing the value or the result of a query,
one must use the value of each subquery, not its result.
\end{remark}

\subsection{Some examples}\  
\label{ssec:spa-examples}

In the examples we assume, as in RDF, that the set $\I$ is the disjoint union
of the set of IRIs and the set of literals, 
where the literals are strings, integers or boolean values.
The literals can be combined by the usual operations on
strings, integers and booleans.

We choose a concrete syntax that is similar to 
the syntax of $\sparql$. For instance a set of triples
$\{(s_1,p_1,o_1),(s_2,p_2,o_2)\}$ is written
\verb+ s1 p1 o1 . s2 p2 o2 . + and braces \verb+{ }+ are used
instead of parentheses $(\;)$.
The evaluation of a query $Q=\construct{R}{P}$ 
is illustrated as in Figure~\ref{fig:sem-construct},  
where each set of mappings $\ulm$ is described by its table  $T(\ulm)$,
as in Remark~\ref{rem:frame-table}: 

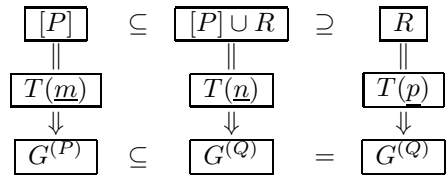
$$ \begin{array}{ccccc}
\begin{array}{|l|}
\hline
\se{P} \\
\hline
\end{array}
& \subseteq & 
\begin{array}{|l|}
\hline
\se{P}\cup R  \\
\hline
\end{array}
& \supseteq &
\begin{array}{|l|}
\hline
R \\ 
\hline
\end{array}
\\ \parallel &&\parallel && \parallel \\ 
\begin{array}{|l|}
  \hline
T(\ulm) \\
\hline
\end{array}
&&
\begin{array}{|l|l|l|}
  \hline
T(\uln) \\
\hline
\end{array}
&&
\begin{array}{|l|l|}
  \hline
T(\ulp) \\
  \hline
\end{array}
\\ \Downarrow && \Downarrow && \Downarrow \\ 
\begin{array}{|l|}
  \hline
\seg{P}{G} \\
  \hline
\end{array}
& \subseteq &
\begin{array}{|l|}
  \hline
\seg{Q}{G} \\
  \hline
\end{array}
& = &
\begin{array}{|l|}
  \hline
\seg{Q}{G} \\
  \hline
\end{array}
\end{array}
$$

\setlength\parindent{0em}

\begin{example}[CONSTRUCT]\
\label{ex:construct}

This example shows how blanks are handled,
whether they are in $G$ or in $R$.

\smallskip

\hfill 
\begin{minipage}{11pc}
\begin{Verbatim}[frame=single,label={Data G},fontsize=\footnotesize,framerule=1.2pt]
_:a employeeName "Alice" . 
_:a employeeId   12345 .
_:b employeeName "Bob" . 
_:b employeeId   67890 .
\end{Verbatim}
\end{minipage}
\hfill 
\begin{minipage}{12pc}
\begin{Verbatim}[frame=single,label={Query Q},fontsize=\footnotesize,framerule=1.2pt]
CONSTRUCT { ?x name _:z }
WHERE { ?x employeeName ?y }
\end{Verbatim}
\end{minipage}
\hfill\null

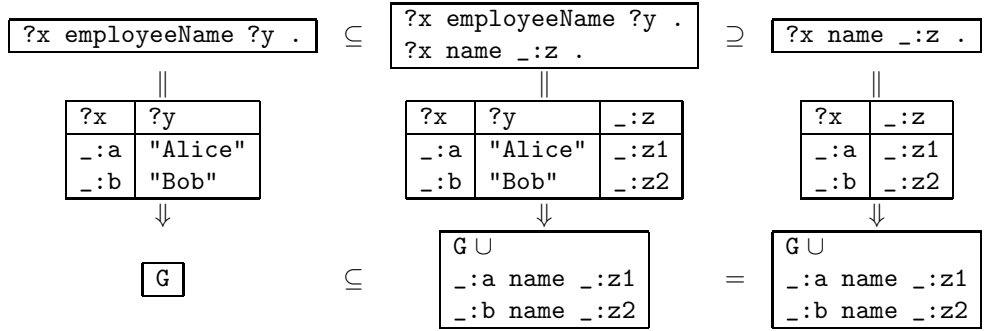
\medskip

$$ \begin{array}{ccccc}
\begin{array}{|l|}
\hline
\verb+?x employeeName ?y .+ \\
\hline
\end{array}
& \subseteq & 
\begin{array}{|l|}
\hline
\verb+?x employeeName ?y .+ \\
\verb+?x name _:z .+ \\
\hline
\end{array}
& \supseteq &
\begin{array}{|l|}
\hline
\verb+?x name _:z .+ \\
\hline
\end{array}
\\ \parallel &&\parallel && \parallel \\ 
\begin{array}{|l|l|}
  \hline
\verb+?x+  & \verb+?y+ \\
  \hline
\verb+_:a+ & \verb+"Alice"+\\
\verb+_:b+ & \verb+"Bob"+ \\ 
  \hline
\end{array}
&&
\begin{array}{|l|l|l|}
  \hline
\verb+?x+  & \verb+?y+ & \verb+_:z+ \\
  \hline
\verb+_:a+ & \verb+"Alice"+ & \verb+_:z1+ \\
\verb+_:b+ & \verb+"Bob"+ & \verb+_:z2+ \\ 
  \hline
\end{array}
&&
\begin{array}{|l|l|}
  \hline
\verb+?x+  & \verb+_:z+ \\
  \hline
\verb+_:a+ & \verb+_:z1+ \\
\verb+_:b+ & \verb+_:z2+ \\
  \hline
\end{array}
\\ \Downarrow && \Downarrow && \Downarrow \\ 
\begin{array}{|l|}
  \hline
\verb+G+ \\
  \hline
\end{array}
& \subseteq &
\begin{array}{|l|}
  \hline
\verb+G+ \;\cup \\
\verb+_:a name _:z1+ \\
\verb+_:b name _:z2+ \\
  \hline
\end{array}
& = &
\begin{array}{|l|}
  \hline
\verb+G+ \;\cup \\
\verb+_:a name _:z1+ \\
\verb+_:b name _:z2+ \\
  \hline
\end{array}
\end{array} $$

It follows that the result of $Q$ over $G$
is the data graph $\mathit{Result}(Q,G) $: 
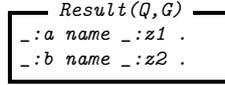
\smallskip

\begin{center}
\begin{minipage}{7pc}
\begin{Verbatim}[frame=single,label={Result(Q,G)},fontsize=\footnotesize,framerule=1.2pt]
_:a name _:z1 .
_:b name _:z2 .
\end{Verbatim}
\end{minipage}
\end{center}

\end{example} 

\begin{example}[SELECT DISTINCT]\
\label{ex:distinct}

A select-distinct-query is equivalent to a construct-query. 

\smallskip

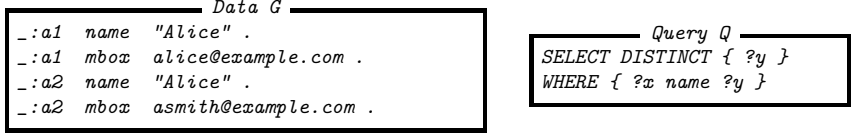
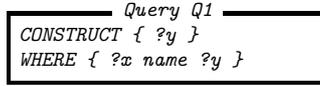
\hfill 
\begin{minipage}{15pc}
\begin{Verbatim}[frame=single,label={Data G},fontsize=\footnotesize,framerule=1.2pt]
_:a1  name  "Alice" .
_:a1  mbox  alice@example.com .
_:a2  name  "Alice" .
_:a2  mbox  asmith@example.com .
\end{Verbatim}
\end{minipage}
\hfill 
\begin{minipage}{10pc}
\begin{Verbatim}[frame=single,label={Query Q},fontsize=\footnotesize,framerule=1.2pt]
SELECT DISTINCT { ?y } 
WHERE { ?x name ?y }
\end{Verbatim}
\end{minipage}
\hfill\null 

\medskip

Equivalent construct-query:

\begin{center}
\begin{minipage}{10pc}
\begin{Verbatim}[frame=single,label={Query Q1},fontsize=\footnotesize,framerule=1.2pt]
CONSTRUCT { ?y } 
WHERE { ?x name ?y }
\end{Verbatim}
\end{minipage}
\end{center}
The value of $Q_1$ over $G$ is computed as in Example~\ref{ex:construct},
it is also the value of $Q$ over $G$: 
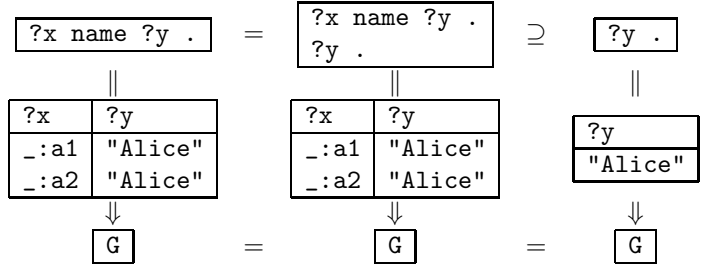
$$ \begin{array}{ccccc}
\begin{array}{|l|}
\hline
\verb+?x name ?y .+ \\
\hline
\end{array}
& = & 
\begin{array}{|l|}
\hline
\verb+?x name ?y .+ \\
\verb+?y .+ \\
\hline
\end{array}
& \supseteq &
\begin{array}{|l|}
\hline
\verb+?y .+ \\
\hline
\end{array}
\\ \parallel &&\parallel && \parallel \\ 
\begin{array}{|l|l|}
  \hline
\verb+?x+  & \verb+?y+ \\
  \hline
\verb+_:a1+ & \verb+"Alice"+\\
\verb+_:a2+ & \verb+"Alice"+\\
  \hline
\end{array}
&&
\begin{array}{|l|l|}
  \hline
\verb+?x+  & \verb+?y+ \\
  \hline
\verb+_:a1+ & \verb+"Alice"+\\
\verb+_:a2+ & \verb+"Alice"+\\
  \hline
\end{array}
&&
\begin{array}{|l|}
  \hline
\verb+?y+ \\
  \hline
\verb+"Alice"+\\
  \hline
\end{array}
\\ \Downarrow && \Downarrow && \Downarrow \\ 
\begin{array}{|l|}
  \hline
\verb+G+ \\
  \hline
\end{array}
& = &
\begin{array}{|l|}
  \hline
\verb+G+ \\
  \hline
\end{array}
& = &
\begin{array}{|l|}
  \hline
\verb+G+ \\
  \hline
\end{array}
\end{array} $$

It follows that the result of $Q$ over $G$
is the set of mappings with table: 
\smallskip

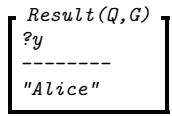
\begin{center}
\begin{minipage}{5pc}
\begin{Verbatim}[frame=single,label={Result(Q,G)},fontsize=\footnotesize,framerule=1.2pt]
?y
--------
"Alice"
\end{Verbatim}
\end{minipage}
\end{center}

\end{example} 

\begin{example}[SELECT]\
\label{ex:select}

A select-query is equivalent to a construct-query, using the query graph $\gr(S)$.

The data graph $G$ is the same as in Example~\ref{ex:distinct} and
the query $Q$ is a select-query, equivalent to the construct-query $Q_1$: 

\smallskip
\hfill 
\begin{minipage}{10pc}
  \begin{Verbatim}[frame=single,label={Query Q},fontsize=\footnotesize,framerule=1.2pt]
SELECT ?y
WHERE { ?x name ?y }
\end{Verbatim}
\end{minipage}
\hfill 
\begin{minipage}{10pc}
  \begin{Verbatim}[frame=single,label={Query Q1},fontsize=\footnotesize,framerule=1.2pt]
CONSTRUCT { _:s py ?y } 
WHERE { ?x name ?y }
\end{Verbatim}
\end{minipage}
\hfill\null

\smallskip
The value of $Q_1$ over $G$ is computed as in Examples~\ref{ex:construct}
and~\ref{ex:distinct},  
it is the value of $Q$ over $G$: 

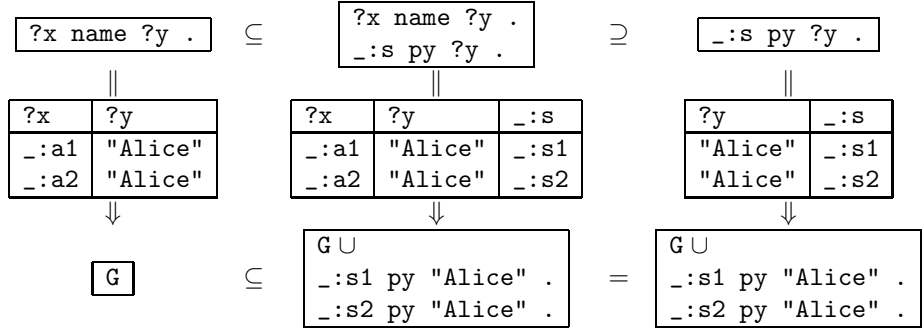
$$ \begin{array}{ccccc}
\begin{array}{|l|}
\hline
\verb+?x name ?y .+ \\
\hline
\end{array}
& \subseteq & 
\begin{array}{|l|}
\hline
\verb+?x name ?y .+ \\
\verb+_:s py ?y .+ \\
\hline
\end{array}
& \supseteq &
\begin{array}{|l|}
\hline
\verb+_:s py ?y .+ \\
\hline
\end{array}
\\ \parallel &&\parallel && \parallel \\ 
\begin{array}{|l|l|}
  \hline
\verb+?x+  & \verb+?y+ \\
  \hline
\verb+_:a1+ & \verb+"Alice"+\\
\verb+_:a2+ & \verb+"Alice"+\\
  \hline
\end{array}
&&
\begin{array}{|l|l|l|}
  \hline
\verb+?x+  & \verb+?y+ & \verb+_:s+ \\
  \hline
\verb+_:a1+ & \verb+"Alice"+ & \verb+_:s1+ \\
\verb+_:a2+ & \verb+"Alice"+ & \verb+_:s2+ \\
  \hline
\end{array}
&&
\begin{array}{|l|l|}
  \hline
\verb+?y+ & \verb+_:s+ \\
  \hline
\verb+"Alice"+ & \verb+_:s1+ \\
\verb+"Alice"+ & \verb+_:s2+ \\
  \hline
\end{array}
\\ \Downarrow && \Downarrow && \Downarrow \\ 
\begin{array}{|l|}
  \hline
\verb+G+ \\
  \hline
\end{array}
& \subseteq &
\begin{array}{|l|}
  \hline
\verb+G+ \;\cup \\
\verb+_:s1 py "Alice" .+ \\
\verb+_:s2 py "Alice" .+ \\
  \hline
\end{array}
& = &
\begin{array}{|l|}
  \hline
\verb+G+ \;\cup \\
\verb+_:s1 py "Alice" .+ \\
\verb+_:s2 py "Alice" .+ \\
  \hline
\end{array}
\end{array} $$

It follows that the result of $Q$ over $G$
is the multiset of mappings with table: 
\smallskip

\begin{center}
\begin{minipage}{5pc}
\begin{Verbatim}[frame=single,label={Result(Q,G)},fontsize=\footnotesize,framerule=1.2pt]
?y
--------
"Alice"
"Alice"
\end{Verbatim}
\end{minipage}
\end{center}
\end{example} 

\begin{example}[UNION]\
\label{ex:union}

This example has to be compared with Example~\ref{ex:union-all}.

\medskip

\hfill 
\begin{minipage}{4.5pc}
\begin{Verbatim}[frame=single,label={Data G},fontsize=\footnotesize,framerule=1.2pt]
a b c .
\end{Verbatim}
\end{minipage}
\hfill 
\begin{minipage}{16pc}
\begin{Verbatim}[frame=single,label={Query Q},fontsize=\footnotesize,framerule=1.2pt]
SELECT ?x 
WHERE { { ?x ?y ?z } UNION { ?x ?y ?z } }
\end{Verbatim}
\end{minipage}
\hfill\null

\medskip

Definition~\ref{def:spa-sem-pattern} implies that
$P {\rm \;UNION\; }P \equiv P$ for any basic pattern $P$,
so that here the query $Q$ is equivalent to $Q_1$: 

\begin{center}
\begin{minipage}{9pc}
\begin{Verbatim}[frame=single,label={Query Q1},fontsize=\footnotesize,framerule=1.2pt]
SELECT ?x 
WHERE { ?x ?y ?z }
\end{Verbatim}
\end{minipage}
\end{center}

The evaluation of $Q_1$ over $G$ runs as follows: 

$$ \begin{array}{ccccc}
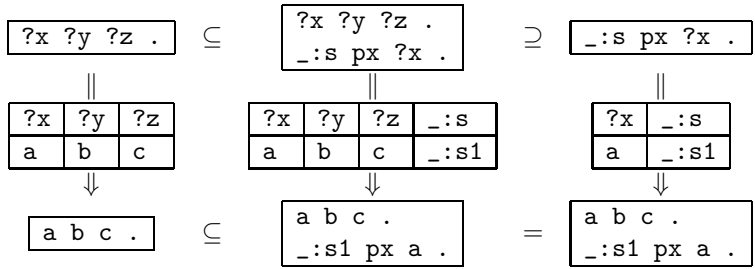
\begin{array}{|l|}
\hline
\verb+?x ?y ?z .+ \\
\hline
\end{array}
& \subseteq & 
\begin{array}{|l|}
\hline
\verb+?x ?y ?z .+ \\
\verb+_:s px ?x .+ \\
\hline
\end{array}
& \supseteq &
\begin{array}{|l|}
\hline
\verb+_:s px ?x .+ \\
\hline
\end{array}
\\ \parallel &&\parallel && \parallel \\ 
\begin{array}{|l|l|l|}
  \hline
\verb+?x+  & \verb+?y+ & \verb+?z+ \\
  \hline
\verb+a+ & \verb+b+ & \verb+c+  \\ 
  \hline
\end{array} 
&&
\begin{array}{|l|l|l|l|}
  \hline
\verb+?x+  & \verb+?y+ & \verb+?z+ & \verb+_:s+ \\
  \hline
\verb+a+ & \verb+b+ & \verb+c+  & \verb+_:s1+ \\ 
  \hline
\end{array} 
&&
\begin{array}{|l|l|}
  \hline
\verb+?x+  & \verb+_:s+ \\
  \hline
\verb+a+ & \verb+_:s1+ \\ 
  \hline
\end{array} 
\\ \Downarrow && \Downarrow && \Downarrow \\ 
\begin{array}{|l|}
  \hline
\verb+a b c .+ \\
  \hline
\end{array}
& \subseteq &
\begin{array}{|l|}
  \hline
\verb+a b c .+ \\
\verb+_:s1 px a .+ \\
  \hline
\end{array}\begin{tabular}{|l|l|l|}
  \hline
\end{tabular}
& = &
\begin{array}{|l|}
  \hline
\verb+a b c .+ \\
\verb+_:s1 px a .+ \\
  \hline
\end{array}
\end{array}
$$

\medskip

Thus, the result of $Q$ over $G$ 
is the multiset of mappings with table: 
\begin{center}
\begin{minipage}{3pc}
\begin{Verbatim}[frame=single,label={Result},fontsize=\footnotesize,framerule=1.2pt]
?x
------
a
\end{Verbatim}
\end{minipage}
\end{center}

\end{example}

\begin{example}[UNION ALL]\
\label{ex:union-all}

This example has to be compared with Example~\ref{ex:union}.

\medskip

\hfill 
\begin{minipage}{4.5pc}
\begin{Verbatim}[frame=single,label={Data G},fontsize=\footnotesize,framerule=1.2pt]
a b c .
\end{Verbatim}
\end{minipage}
\hfill 
\begin{minipage}{18pc}
\begin{Verbatim}[frame=single,label={Query Q},fontsize=\footnotesize,framerule=1.2pt]
SELECT ?x 
WHERE { { ?x ?y ?z } UNION ALL { ?x ?y ?z } }
\end{Verbatim}
\end{minipage}
\hfill\null

As in Remark~\ref{rem:spa-query-union}
this means that the query $Q$ is equivalent to $Q_1$:

\smallskip

\hfill 
\begin{minipage}{18pc}
\begin{Verbatim}[frame=single,label={Query Q1},fontsize=\footnotesize,framerule=1.2pt]
SELECT ?x 
WHERE {
  { SELECT { ?x ?y ?z } WHERE { ?x ?y ?z } }
  UNION 
  { SELECT { ?x ?y ?z } WHERE { ?x ?y ?z } } }
\end{Verbatim}
\end{minipage}
\hfill\null

\medskip

Here the pattern $P=\select{\{?x ?y ?z\}}{\{?x ?y ?z\}}$ is not basic,
and we now check that in fact $P {\rm \;UNION\; }P $ is not equivalent to $P$. 

The following diagram illustrates the value of $P$ over $G$, 
then the value of $P$ over $\seg{P}{G}$ and 
finally their union as sets of mappings, which is the value of
$P {\rm \;UNION\; }P $ over $G$. 
The difference between $\sem{P}{G}$ and $\sem{P}{\seg{P}{G}}$
is that the value of \verb+_:s+ must be a \emph{fresh} blank,
so that once some blank, say \verb+_:s1+, is chosen for 
$\sem{P}{G}$ then another blank, say \verb+_:s2+, must be chosen for
$\sem{P}{\seg{P}{G}}$.

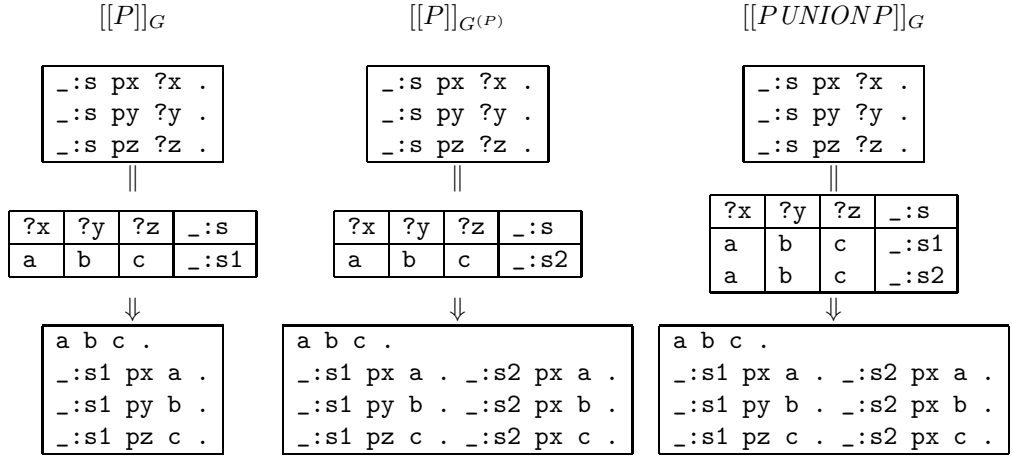
$$ 
\begin{array}{ccc}
\sem{P}{G}
&
\sem{P}{\seg{P}{G}}
&
\sem{P\mathit{ UNION }P}{G}
\\
&& \\
\begin{array}{|l|}
\hline
\verb+_:s px ?x .+ \\ 
\verb+_:s py ?y .+ \\ 
\verb+_:s pz ?z .+ \\ 
\hline
\end{array}
&
\begin{array}{|l|}
\hline
\verb+_:s px ?x .+ \\ 
\verb+_:s py ?y .+ \\ 
\verb+_:s pz ?z .+ \\ 
\hline
\end{array}
&
\begin{array}{|l|}
\hline
\verb+_:s px ?x .+ \\ 
\verb+_:s py ?y .+ \\ 
\verb+_:s pz ?z .+ \\ 
\hline
\end{array}
\\ 
\parallel & \parallel & \parallel
\\ 
\begin{array}{|l|l|l|l|}
  \hline
\verb+?x+ & \verb+?y+ & \verb+?z+ & \verb+_:s+ \\
  \hline
\verb+a+ & \verb+b+ & \verb+c+ & \verb+_:s1+ \\
  \hline
\end{array}
&
\begin{array}{|l|l|l|l|}
  \hline
\verb+?x+ & \verb+?y+ & \verb+?z+ & \verb+_:s+ \\
  \hline
\verb+a+ & \verb+b+ & \verb+c+ & \verb+_:s2+ \\
  \hline
\end{array}
&
\begin{array}{|l|l|l|l|}
  \hline
\verb+?x+ & \verb+?y+ & \verb+?z+ & \verb+_:s+ \\
  \hline
\verb+a+ & \verb+b+ & \verb+c+ & \verb+_:s1+ \\
\verb+a+ & \verb+b+ & \verb+c+ & \verb+_:s2+ \\
  \hline
\end{array}
\\
\Downarrow & \Downarrow & \Downarrow
\\ 
\begin{array}{|l|}
\hline
\verb+a b c .+ \\ 
\verb+_:s1 px a .+ \\ 
\verb+_:s1 py b .+ \\ 
\verb+_:s1 pz c .+ \\ 
  \hline
\end{array}
&
\begin{array}{|l|}
  \hline
\verb+a b c .+ \\ 
\verb+_:s1 px a . + \verb+_:s2 px a .+ \\ 
\verb+_:s1 py b . + \verb+_:s2 px b .+ \\ 
\verb+_:s1 pz c . + \verb+_:s2 px c .+ \\ 
  \hline
\end{array}
&
\begin{array}{|l|}
  \hline
\verb+a b c .+ \\ 
\verb+_:s1 px a . + \verb+_:s2 px a .+ \\ 
\verb+_:s1 py b . + \verb+_:s2 px b .+ \\ 
\verb+_:s1 pz c . + \verb+_:s2 px c .+ \\ 
  \hline
\end{array}
\\
\end{array}
$$

Finally, by projecting on \verb+?x+, 
the result of $Q$ over $G$ is the multiset of mappings with table: 

\begin{center}
\begin{minipage}{3pc}
\begin{Verbatim}[frame=single,label={Result},fontsize=\footnotesize,framerule=1.2pt]
?x
------
a
a
\end{Verbatim}
\end{minipage}
\end{center}

\end{example} 
\begin{example}[EXISTS]\
\label{ex:exists}

This example is based on Example~4.6 in \cite{KKC}. 
The query in \cite{KKC} is similar to the query $Q_0$ below,
however in the language $\spa$
this query is not syntactically valid since $\V(\mbox{BOUND}(?x))$
is not included in $\V(\{(?y, ?y ?y)\})$.
A valid query $Q$ is obtained by shifting braces. 

\medskip
  
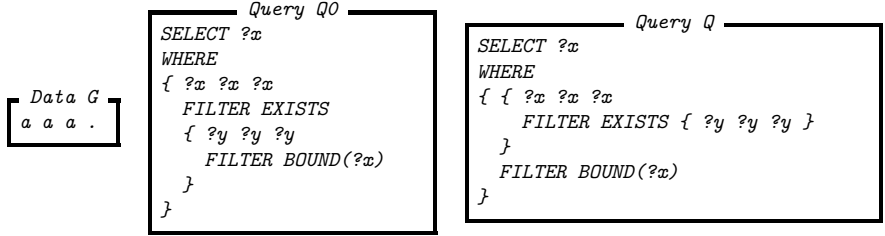
\hfill 
\begin{minipage}{3.5pc}
\begin{Verbatim}[frame=single,label={Data G},fontsize=\footnotesize,framerule=1.2pt]
a a a .
\end{Verbatim}
\end{minipage}
\hfill 
\begin{minipage}{9pc}
\begin{Verbatim}[frame=single,label={Query Q0},fontsize=\footnotesize,framerule=1.2pt]
SELECT ?x 
WHERE
{ ?x ?x ?x
  FILTER EXISTS
  { ?y ?y ?y 
    FILTER BOUND(?x)
  }  
}
\end{Verbatim}
\end{minipage}
\hfill 
\begin{minipage}{13pc}
\begin{Verbatim}[frame=single,label={Query Q},fontsize=\footnotesize,framerule=1.2pt]
SELECT ?x 
WHERE
{ { ?x ?x ?x
    FILTER EXISTS { ?y ?y ?y }
  }
  FILTER BOUND(?x) 
}
\end{Verbatim}
\end{minipage}
\hfill\null

\smallskip

Thus $Q=\select{?x}{\{\,\{P_1 {\rm\;FILTER\; EXISTS\;} P_2\} {\rm\;FILTER\;}
{\rm\;BOUND\;} (?x)\,\}}$
with $P_1=\{(?x, ?x, ?x)\}$ and $P_2=\{(?y, ?y, ?y)\}$.
The unique mapping $m_1$ from $P_1$ to $G$ is such that $m_1(?x)=a$
and the unique mapping $m_2$ from $P_2$ to $G$ is such that $m_2(?y)=a$,
they are compatible, so that the value of
$P_1 {\rm\;FILTER\; EXISTS\;} P_2$ over $G$ is $\{m_1\}$.
Since $m_1$ binds $?x$ to $a$, the value of the expression 
${\rm BOUND\;} (?x)$ is $\true$, thus the value of $Q$ over $G$ is
$\sem{Q}{G}=\{m_1\}: \{(?x, ?x, ?x)\} \To \{(a,a,a)\} $.

\end{example} 

\begin{example}[subquery]\
\label{ex:sub}

This example is based on the example following Example~4.6 in \cite{KKC}.

\medskip
  
\hfill 
\begin{minipage}{3.5pc}
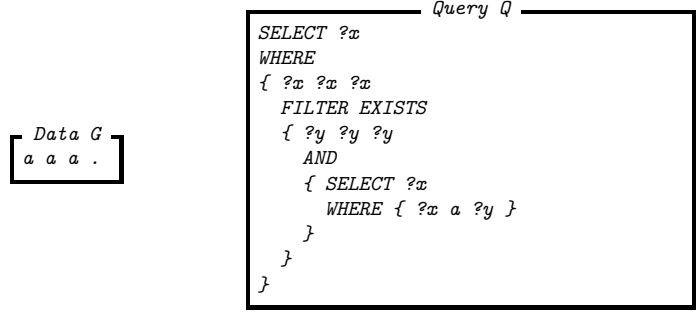
\begin{Verbatim}[frame=single,label={Data G},fontsize=\footnotesize,framerule=1.2pt]
a a a .
\end{Verbatim}
\end{minipage}
\hfill 
\begin{minipage}{14pc}
\begin{Verbatim}[frame=single,label={Query Q},fontsize=\footnotesize,framerule=1.2pt]
SELECT ?x 
WHERE
{ ?x ?x ?x
  FILTER EXISTS
  { ?y ?y ?y
    AND
    { SELECT ?x
      WHERE { ?x a ?y }
    }
  }
}
\end{Verbatim}
\end{minipage}
\hfill\null

\smallskip 
 
Here $Q = \select{?x}{P}$
with $P=P_1 \mbox{ FILTER EXISTS } P_2$
\\ and $P_2 = P_3 \mbox{ AND } P_4$ 
where $P_4$ is a select-query. 

Since $P_1$ is a basic pattern $\sem{P_1}{G}:P_1\To G$ is such that
$T(\sem{P_1}{G}) = \begin{tabular}{|l|}
  \hline
\verb+?x+  \\
  \hline
\verb+a+ \\ 
  \hline
\end{tabular} $

Similarly $\sem{P_3}{G}:P_3\To G$ is such that
$T(\sem{P_3}{G}) = \begin{tabular}{|l|}
  \hline
\verb+?y+  \\
  \hline
\verb+a+ \\ 
  \hline
\end{tabular} $

The value of query $P_4$ over $G$ is 
$\sem{P_4}{G}:\{(\blank{s},p_x,?x)\}\To G\cup\{(\blank{s_1},p_x,a)\}$
\\ such that (as in Example~\ref{ex:select}) 
$T(\sem{P_4}{G}) = \begin{tabular}{|l|l|}
  \hline
\verb+?x+ & \verb+_:s+  \\
  \hline
\verb+a+ & \verb+_:s1+  \\ 
  \hline
\end{tabular} $

Thus $\sem{P_2}{G}:P_2\To G$ is such that
$T(\sem{P_2}{G}) = \begin{tabular}{|l|l|l|}
  \hline
\verb+?x+ & \verb+?y+ & \verb+_:s+  \\
  \hline
\verb+a+ & \verb+a+ & \verb+_:s1+  \\ 
  \hline
\end{tabular} $

and finally $\sem{P}{G}:P_1 \To G$ is such that
$T(\sem{P}{G}) = \begin{tabular}{|l|}
  \hline
\verb+?x+ \\
  \hline
\verb+a+ \\ 
  \hline
\end{tabular} $

When this table is seen as a multiset of mappings, 
it is the result of $Q$ over $G$: 
\begin{center}
\begin{minipage}{3pc}
\begin{Verbatim}[frame=single,label={Result},fontsize=\footnotesize,framerule=1.2pt]
?x
------
a
\end{Verbatim}
\end{minipage}
\end{center}

\end{example} 

\begin{example}[assignment]\
\label{ex:bind}

This example is based on Example~5.4 in \cite{KKC}.

\medskip
  
\hfill 
\begin{minipage}{3.5pc}
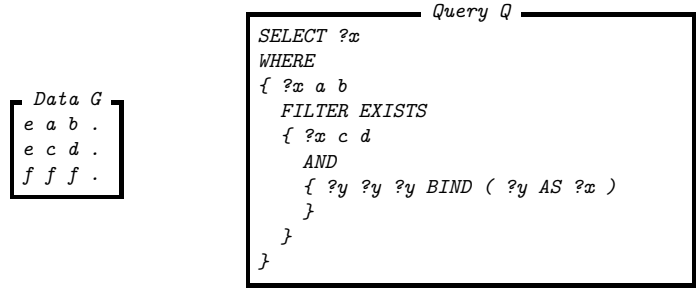
\begin{Verbatim}[frame=single,label={Data G},fontsize=\footnotesize,framerule=1.2pt]
e a b .
e c d .
f f f .
\end{Verbatim}
\end{minipage}
\hfill 
\begin{minipage}{14pc}
\begin{Verbatim}[frame=single,label={Query Q},fontsize=\footnotesize,framerule=1.2pt]
SELECT ?x 
WHERE
{ ?x a b 
  FILTER EXISTS
  { ?x c d 
    AND 
    { ?y ?y ?y BIND ( ?y AS ?x )
    }
  }
} 
\end{Verbatim}
\end{minipage}
\hfill\null

\smallskip

Note that the query $Q$ is syntactically correct since
$?x\not\in \V(\{(?y,?y,?y)\})$. 
The main point in the evaluation of $Q$ over $G$ is that the subexpression
EXISTS $P$ evaluates to $\false$, as explained below. Then clearly the result
of the query is the empty multiset of mappings.
The unique mapping in $\sem{\{(?x,c,d)\}}{G}$ sends $?x$ to $e$.
The unique mapping in $\sem{\{(?y,?y,?y)\}}{G}$ sends $?y$ to $f$
then BIND ($?y$ AS $?x$) extends this mapping by sending $?x$ to $f$.
Thus the mappings are not compatible and the join is the empty
set of mappings, as required. 

\end{example} 

\setlength\parindent{1.5em}

\section{Conclusion}
\label{conclusion}

We proposed a core language GrAL close to SPARQL for which we
proposed a uniform semantics. This semantics allows one
to compose different queries and patterns regardless the different
forms of the queries. In this paper we did not include all
SPARQL query forms such as ASK or DESCRIBE, nor did we mention aggregates
or so. We intend to include such SPARQL features in a forthcoming report.
The proposed framework has been illustrated on RDF
graphs and SPARQL queries but it is tailored to fit any kind of graph
structures with a clear notion of graph homomorphism, see e.g., the different
structures mentioned in \cite{AnglesABHRV17}. Coming back to the title
of the paper, which might be a bit provocative, it emphasizes on a
feature of our semantics which makes it possible to encode easily any
SELECT or SELECT DISTINCT query as a CONSTRUCT query.

\bibliographystyle{plain}
\bibliography{biblio}

\begin{thebibliography}{10}

\bibitem{AnglesABHRV17}
Renzo Angles, Marcelo Arenas, Pablo Barcel{\'{o}}, Aidan Hogan, Juan~L.
  Reutter, and Domagoj Vrgoc.
\newblock Foundations of modern query languages for graph databases.
\newblock {\em {ACM} Comput. Surv.}, 50(5):68:1--68:40, 2017.

\bibitem{AnglesG11}
Renzo Angles and Claudio Guti{\'{e}}rrez.
\newblock Subqueries in {SPARQL}.
\newblock In Pablo Barcel{\'{o}} and Val Tannen, editors, {\em Proceedings of
  the 5th Alberto Mendelzon International Workshop on Foundations of Data
  Management, Santiago, Chile, May 9-12, 2011}, volume 749 of {\em {CEUR}
  Workshop Proceedings}. CEUR-WS.org, 2011.

\bibitem{KKC}
Mark Kaminski, Egor~V. Kostylev, and Bernardo {Cuenca Grau}.
\newblock Query nesting, assignment, and aggregation in {SPARQL} 1.1.
\newblock {\em {ACM} Trans. Database Syst.}, 42(3):17:1--17:46, 2017.

\bibitem{Kim82}
Won Kim.
\newblock On optimizing an sql-like nested query.
\newblock {\em {ACM} Trans. Database Syst.}, 7(3):443--469, 1982.

\bibitem{Kostylevetal2015}
Egor~V. Kostylev, Juan~L. Reutter, and Mart{\'{\i}}n Ugarte.
\newblock {CONSTRUCT} queries in {SPARQL}.
\newblock In {\em 18th International Conference on Database Theory, {ICDT}
  2015, March 23-27, 2015, Brussels, Belgium}, pages 212--229, 2015.

\bibitem{PerezAG09}
Jorge P{\'{e}}rez, Marcelo Arenas, and Claudio Guti{\'{e}}rrez.
\newblock Semantics and complexity of {SPARQL}.
\newblock {\em {ACM} Trans. Database Syst.}, 34(3):16:1--16:45, 2009.

\bibitem{PolleresRK16}
Axel Polleres, Juan~L. Reutter, and Egor~V. Kostylev.
\newblock Nested constructs vs. sub-selects in {SPARQL}.
\newblock In Reinhard Pichler and Altigran~Soares da~Silva, editors, {\em
  Proceedings of the 10th Alberto Mendelzon International Workshop on
  Foundations of Data Management, Panama City, Panama, May 8-10, 2016}, volume
  1644 of {\em {CEUR} Workshop Proceedings}. CEUR-WS.org, 2016.

\bibitem{Robinson2013}
Ian Robinson, Jim Webber, and Emil Eifrem.
\newblock {\em Graph Databases}.
\newblock O'Reilly Media, Inc., 2013.

\bibitem{RNpropertyGraphs}
Marko~A. Rodriguez and Peter Neubauer.
\newblock Constructions from dots and lines.
\newblock {\em CoRR}, abs/1006.2361, 2010.

\bibitem{Schmidt0L10}
Michael Schmidt, Michael Meier, and Georg Lausen.
\newblock Foundations of {SPARQL} query optimization.
\newblock In Luc Segoufin, editor, {\em Database Theory - {ICDT} 2010, 13th
  International Conference, Lausanne, Switzerland, March 23-25, 2010,
  Proceedings}, {ACM} International Conference Proceeding Series, pages 4--33.
  {ACM}, 2010.

\bibitem{sparql}
{SPARQL 1.1 Query Language}.
\newblock W3C Recommendation, march 2013.

\bibitem{rdf}
{RDF 1.1 Concepts and Abstract Syntax}.
\newblock W3C Recommendation, February 2014.

\end{thebibliography}

\end{document}